\documentclass[preprint,authoryear,12pt]{elsnote}
\usepackage{graphicx}
\usepackage{tikz}
\usepackage{pgfplots}

\usepackage{float}
\usepackage{amssymb}
\usepackage{amsthm}

\usepackage{booktabs}

\usepackage{amsmath}

\newtheorem{thm}{Theorem}
\newtheorem{lem}[thm]{Lemma}
\newtheorem{prop}[thm]{Proposition}
\newdefinition{rmk}{Remark}
\newdefinition{defn}{Definition}
\newdefinition{eg}{Example}

\begin{document}

\begin{frontmatter}



\title{Price of Anarchy for Non-atomic Congestion Games with Stochastic Demands}


\author{Chenlan Wang}
\ead{chenlan.wang@warwick.ac.uk}
\author{Xuan Vinh Doan}
\ead{X.Doan@warwick.ac.uk}
\author{Bo Chen\corref{cor1}}
\ead{B.Chen@warwick.ac.uk}
\cortext[cor1]{Corresponding author. Tel.: +44 2476524755}

\address{Warwick Business School $\&$  Centre for Discrete Mathematics and its Applications (DIMAP), University of Warwick, CV4 7AL, UK}

\begin{abstract}
We generalize the notions of user equilibrium and system optimum to non-atomic congestion games with stochastic demands.
We establish upper bounds on the price of anarchy for three different settings of link cost functions and demand distributions,
namely, (a) affine cost functions and general distributions, (b) polynomial cost functions and general positive-valued
distributions, and (c) polynomial cost functions and the normal distributions. All the upper bounds are tight in some special cases,
including the case of deterministic demands.
\end{abstract}


\end{frontmatter}
\section{Introduction}
\label{}

Nonatomic congestion games illustrate non-cooperative situations involving large populations of players competing for a finite set of resources \Citep{Chau}. Routing problem in transportation networks is a very important application of non-atomic congestion games.
The price of anarchy (PoA), first introduced by \Citet{Koutsoupias} on a load-balancing game, is one of the main measures of system degradation due to lack of coordination. \Citet{Roughgarden2004} studied the PoA for non-atomic congestion games as the worst-case performance of the user equilibrium (UE) in terms of system optimality achieved at the system optimum (SO), where the UE \Citep{Wardrop} describes a steady state of travelers' selfish routing while the SO demonstrates the optimal usage of traffic resources as a result of a well-coordinated action on the whole network.

Quantitative study on the PoA enables us to deem certain outcomes of a game optimal or approximately optimal and to make known the factor influencing the inefficiency of the UE, and further contributes to mechanism design for congestion games. \Citet{Roughgarden2002} bounded the PoA when the link cost functions are separable, semi-convex and differentiable. The PoA was proved dependent only on the class of the cost functions, independent of the network topology. In particular, the PoA with affine cost functions is tightly bounded by $4/3$.

The main developments in the research on PoA were extensions to networks with a broader range of cost functions. \Citet{Chau} generalized Roughgarden and Tardos' results to the cases with symmetric cost functions. \Citet{Correa2004} gave a new proof of the upper bound of the PoA with cost functions that are non-convex, non-differentiable, and even discontinuous. \Citet{Perakis} extended the work to asymmetric cost functions and bounded the PoA by two parameters of asymmetry and nonlinearity. \Citet{Sheffi} introduced the notion of stochastic user equilibrium (SUE), which describes the travelers' selfish routing decisions on their subjective perceived travel costs by involving stochastic cost functions. The PoA on logit-based SUE was bounded by \Citet{Guo2010} on the basis of Sheffi's model.

Another line of developments in the PoA study is to improve the setting of the traffic demand to better reflect reality. \Citet{Chau} presented a weaker upper bound on the PoA with elastic demands. Although study on the PoA with stochastic demands is still quite new, efforts have been spent on modelling UE and SO involving demand uncertainty. It was assumed that the objective of selfish travelers was to choose the route that minimizes the mean travel cost \Citep{Sumalee2011} or weighted sum of the mean and the variance of the travel cost \Citep{Sumalee2011, Bell2002} with risk-neutral and risk-averse travelers, respectively. A travel time budget (TTB) was also considered in the equilibrium condition on the basis of reliability \Citep{Lo2006,Shao2006}. However, to deduce the distributions of the path and link flows, all these studies rely on some assumptions, such as that all the path flows follow the same type of distribution as the demand and have the same variance (or standard deviation) to mean ratio \Citep{Sumalee2011,Shao2006,Zhou2008}, and that all the path flows are independent \Citep{Watling2005,Sumalee2011,Shao2006,Zhou2008}. Apparently, these assumptions are open to questions for the relationship between the path flows and demand, not only because of lack of empirical data support but also they violate the demand feasibility constraint even in simple networks. In order to have a better reliable result on the PoA, we need to relax the aforementioned assumptions and establish a new equilibrium condition.

In this paper we present an analytical method to determine distributions of the path and link flows under a given demand distribution and, from a practical perspective, describe travelers' behaviors by route choice probabilities. We generalize the deterministic UE condition to a stochastic version with risk-neutral travelers. For our new model we establish upper bounds on the PoA, which are found to depend on cost functions and demand distributions. All these upper bounds are shown  to be tight in some special cases.

The remainder of the paper is organized as follows. Section~2 introduces generalized notions of user equilibrium (UE) and system optimum (SO) under demand uncertainty, formulates the equilibrium condition as a variational inequality problem and discusses existence and uniqueness of the equilibrium. Section~3 studies the PoA with affine cost functions and polynomial cost functions respectively. For polynomial cost functions, we first present an upper bound on the PoA for any general positive-valued demand distribution. Then we improve the upper bound when the demand distribution is specifically normal. In Section~4 we compare the upper bounds established in Section~3 and discuss connections with existing results in the literature. Conclusions are drawn in Section~5.

\section{Model with stochastic demand}

\subsection{The route choice model}
\label{sec:with_notation}

Consider a general network $G=(N,E)$, where $N$ and $E$ denote the set of nodes and edges, receptively. A subset of nodes forms a set of origin-destination (O-D) pairs, denoted by $I$. We call an O-D pair $i\in I$ a \emph{commodity}. Parallel edges are allowed and a node can be in multiple O-D pairs. Denote by $P_i$ the set of all possible paths connecting an O-D pair $i\in I$.

Day-to-day variability of the traffic demand is considered as the source of the uncertainty in this study. We assume that the demand distributions are given and publicly known, which is based on the fact that a traveler, especially a commuter, has knowledge of the probabilities of possible demand levels from his or her own experiences, although the actual current demand level is unknowable. A similar assumption can be found in the model of deterministic demands, which states that travelers have perfect knowledge of the fixed demand in the network \Citep{Wardrop}. The demands of different O-D pairs are assumed to be independent. We adopt the following notation in our study, where capital letters and lower cases letters are used to express random variables and, if applicable, the corresponding mean values, respectively.

\begin{itemize}
\item[$\mathbf{D}$: ] vector of traffic demands with component $D_{i}\in\mathbb{R}$ as the demand between O-D pair $i\in I$; \vspace{-10pt}
\item[$\mathbf{d}$: ] vector of mean traffic demands with compoent $d_i\in\mathbb{R}$ as the mean demand between O-D pair $i\in I$; \vspace{-10pt}
\item[$\sigma_i^2$: ] variance of $D_i$; \vspace{-10pt}
\item[$\theta_i$: ] coefficient of demand variation, i.e., $\theta_i=\sigma_i/d_i$; \vspace{-10pt}
\item[$\overline{\theta}$: ] maximum coefficient of demand variation, i.e., $\overline{\theta}=\max_{i\in I} \{\theta_i\}$; \vspace{-10pt}
\item[$\underline{\theta}$: ] minimum coefficient of demand variation, i.e., $\underline{\theta}=\min_{i\in I} \{\theta_i\}$; \vspace{-10pt}
\item[$F_{k}^{i}$: ] traffic flow on path $k\in P_{i}$, $F_{k}^{i} \in \mathbb{R}$; \vspace{-10pt}
\item[$f_{k}^{i}$: ] mean flow on path $k\in P_{i}$, $f_{k}^{i} \in \mathbb{R}$; \vspace{-10pt}
\item[$\mathbf{F}$: ] vector of path flows, i.e., $\mathbf{F}=(F_k^i: k\in P_i, i\in I)$; \vspace{-10pt}
\item[$\mathbf{f}$: ] vector of mean path flows, i.e., $\mathbf{f}=(f_k^i: k\in P_i, i\in I)$; \vspace{-10pt}
\item[$V_{e}$: ] traffic flow on edge $e\in E$, $V_{e} \in \mathbb{R}$; \vspace{-10pt}
\item[$v_{e}$: ] mean traffic flow on edge $e\in E$, $v_{e} \in \mathbb{R}$; \vspace{-10pt}
\item[$\mathbf{V}$: ] vector of link flows, i.e., $\mathbf{V}=(V_e: e\in E)$; \vspace{-10pt}
\item[$\mathbf{v}$: ] vector of mean link flows, i.e., $\mathbf{v}=(v_e: e\in E)$; \vspace{-10pt}
\item[$\delta_{k,e}^i$: ] link-path incidence indicator, which is 1 if link $e$ is included in path $k\in P_i$ and 0 otherwise, $e\in E,\ i\in I$; \vspace{-10pt}
\item[$\delta_e^i$: ] link-commodity incidence indicator, i.e., $\delta_e^i=\max_{k\in P_i}\delta_{k,e}^i$; \vspace{-10pt}
\item[$n_e$: ] number of O-D pairs that use link $e\in E$ in their paths, i.e., $n_e = \sum_{i\in I}\delta_e^i$;
    \vspace{-10pt}
\item[$n$: ] $n=\max_{e\in E} \{n_e\}$. Hence $n \le |I|$.
\end{itemize}

Given stochastic demand vector $\mathbf{D}=(D_i: i\in I)$, a multi-commodity flow $\mathbf{F}=(F_k^i: k\in P_i, i\in I)$ is said to be \emph{feasible} if
\begin{equation}\label{eqn:feasible_flow}
    \sum_{k\in P_i} F_k^i = D_i, \quad \forall\ i\in I.
\end{equation}
It is clear that the flow on each link is the sum of flows on all the paths that include the link:
\begin{equation}\label{eqn:Ve}
V_e=\sum_{i\in I}\underset{k\in P_i}\sum \delta_{k,e}^i F_k^i, \quad \forall\ e\in E.
\end{equation}

In a non-atomic congestion game, there are an infinite number of players, each controlling a negligible fraction of the overall traffic. The cost, denoted by $c_e(\cdot)$: $\mathbb{R}^{+}\rightarrow \mathbb{R}^{+}$, of traveling through edge $e\in E$ is a nondecreasing function of the total flow on it, which is also called a (link) cost function. The path cost is simply the sum of the costs of those links that are included in the path, i.e.,
\[ 
c_k^i(\mathbf{F})=\underset{e\in E}\sum \delta_{k,e}^i c_e(V_e),\quad \forall\  k\in P_i, \forall\ i\in I.
\] 
We denote any instance of a non-atomic congestion game by a triple $(G,\mathbf{D},\mathbf{c})$, where $G$ is the underlying network, $\mathbf{D}$ and $\mathbf{c}$ are the vectors of demands and (link) cost functions,
respectively.

Note that pure strategies and mixed strategies are regarded as the same in the deterministic UE model \Citep{Roughgarden2002}, as flow assignments according to mixed strategies can be obtained via pure strategies according to flow proportions. This is based on the assumption that all the other players' behaviors are known when one player makes a route choice. However, this assumption becomes no longer valid under stochastic demands and it is unrealistic to distinguish individual traffic of the same O-D pair at an equilibrium according to different routes taken. Thus it is necessary for us to consider mixed strategies since it is reasonable to assume that all the players of the same O-D pair play the same strategies at an equilibrium in such an environment with incomplete information \Citep{Myerson1998,Ashlagi}.

For any O-D pair $i\in I$, let $p_k^i$ be the probability that path $k\in P_i$ is chosen. Let
\begin{equation*}
\Omega=\{\mathbf{p}=(p_k^i\ge 0: k\in P_i, i\in I): \sum_{k\in P_i} p_k^i=1, i\in I\}.
\end{equation*}
Then $\Omega$ is the set of vectors of route choice probabilities across all the paths with a dimension of $ \sum_{i\in I}|P_i|$. In order to describe the traffic assignment under stochastic demand, we adopt the \textit{route choice model} to simulate travelers' path choice behaviors, which has been widely used in stochastic routing problems \Citep{Sheffi,Watling2005}. It needs to be noted that the route choice probabilities in Sheffi's model are used to describe the routing trend among all the travelers, which are estimated by the flow fractions on each path, and the routing choice of a specific traveler is still a pure strategy, determined by his or her own estimation of the travel cost. On the other hand, the routing choice probabilities in this study are in fact mixed strategies undertaken by travelers. The traffic assignment under stochastic demands is determined by the routing choice probabilities. The path and link flows are random variables related to the random demands and the routing choice probabilities, which consequently induce the random path and link travel costs.

In what follows we show that the distributions of link flows can be identified by the demand distributions and the routing choice probabilities. Since each traveler between any O-D pair $i\in I$ controls a negligible amount of traffic, $\Delta d_i$, the number of travelers on path $k\in P_i$ after demand $D_i$ is realized at $D_i = y$ is
\begin{equation}\label{eqn:path_traveler_number}
    m_k^i=\frac{F_k^i|D_i=y}{\Delta d_i}, \quad\forall\ k\in P_i.
\end{equation}
Since the routing choice on path $k\in P_i$ for every such traveler is a Bernoulli event with success probability $p_k^i$, the conditional number $m_k^i$ of travelers follows a multinomial distribution with the mean and variance as follows:
\begin{align*}
&\mathbb{E}[m_k^i]= m_ip_k^i,  &&\forall\ k\in P_i,\\
&\hbox{Var}[m_k^i] =m_ip_k^i(1-p_k^i), &&\forall\ k\in P_i,
\end{align*}
where $m_i = y/(\Delta d_i)$. Therefore, it follows from \eqref{eqn:path_traveler_number} that
\begin{align*}
&\mathbb{E}[F_k^i|D_i=y]=\Delta d_i m_ip_k^i=y p_k^i, \\
&\hbox{Var}[F_k^i|D_i=y]=(\Delta d_i)^2 \hbox{Var}[m_k^i]= \Delta d_i yp_k^i(1-p_k^i),
\end{align*}
for any $k\in P_i$. The variance above vanishes as $\Delta d_i \rightarrow 0$, which implies
\begin{equation*}
(F_k^i|D_i=y)\cong p_k^i\cdot y, \quad \forall\ k\in P_i,\  \forall\ i\in I.
\end{equation*}
Therefore, the path flows are determined once the demands are realized. The distributions of path flows then follows from the \textit{total probability theorem} as follows:
\begin{equation}\label{eqn:flow_distribution}
F_k^i \cong p_k^i\cdot D_i, \quad \forall\ k\in P_i,\  \forall\ i\in I.
\end{equation}
Similarly we can obtain the distributions of random link flows with the link-path conservation \eqref{eqn:Ve} .

\smallskip\textbf{Remark.} It is commonly assumed in the literature \Citep{Watling2005, Sumalee2011, Shao2006, Zhou2008} that all path flows $\{F_k^i: k\in P_i, i\in I\}$
are independent, which apparently violates the flow feasibility constraints \eqref{eqn:feasible_flow}. On the other hand, if we only assume that $\{D_i: i\in I\}$ are independent, then \eqref{eqn:flow_distribution} implies that, for any $i,i'\in I$ and any $k\in P_i, k'\in P_{i'}$, path flows $F_k^i$ and $F_{k'}^{i'}$ are independent of each other.

\subsection{Equilibrium under stochastic demand (UE-SD)}

As discussed in the previous section, under stochastic traffic demands we assume that risk-neutral travelers between the same O-D pair will use the same strategy at a steady state. We define our equilibrium condition such that
travelers cannot improve their expected travel costs by unilaterally changing their routing choice strategies.
\begin{defn}[UE-SD condition]
Given a transportation game $(G, \mathbf{D}, \mathbf{c})$, vector $\mathbf{p}\in \Omega$ of routing choice
probabilities is said to be a user equilibrium (UE-SD) if and only if
\begin{equation}\label{eqn:def_UE-SD}
\mathbb{E}[c_k^i(\mathbf{F})]\leq \mathbb{E}[c_\ell^i(\mathbf{F})], \quad\forall\ k, \ell\in P_i, i\in I \textrm{ with } p_k^i>0.
\end{equation}
\end{defn}

From the definition we see that, at UE-SD, all the paths with positive probabilities for the same O-D pair have the equal and minimum expected travel cost. When all travelers play mixed strategies according to the UE-SD condition, the expected travel costs are guaranteed to be at minimum. To solve the equilibrium problem, let us reformulate the UE-SD condition into a variational inequality (VI).
\begin{prop}
Given a transportation game $(G, \mathbf{D}, \mathbf{c})$, let $\mathbf{p^\ast}\in\Omega$ be a vector of routing choice probabilities. Then $\mathbf{p^\ast}$ is a UE-SD if and only if it satisfies the following VI condition: for any vector $\mathbf{p}$ of routing choice probabilities,
\begin{equation}\label{eqn:UE-SD-VI}
(\mathbf{f}-\mathbf{f}^{\ast})^T \mathbb{E}\left[c(\mathbf{F}^{\ast})\right]\geq 0,
\end{equation}
where $\mathbf{F}^{\ast}$ is the vector of path flows corresponding to $\mathbf{p}^{\ast}$, and $\mathbf{f}^{\ast}$ and $\mathbf{f}$ are, respectively, the vector of the mean path flow corresponding to $\mathbf{p}^{\ast}$ and $\mathbf{p}$.
\end{prop}
\begin{proof}
Since demand $d_i>0$ for every $i\in I$, according to \eqref{eqn:flow_distribution} we can write the UE-SD
condition \eqref{eqn:def_UE-SD} as follows:
\begin{equation}\label{SN-UE-poly-f}
\mathbb{E}[c_k^i(\mathbf{F})]\leq \mathbb{E}[c_\ell^i(\mathbf{F})], \quad\forall\ k, \ell\in P_i, i\in I \textrm{ with } f_k^i>0.
\end{equation}
Let $\pi_i= \min_{\ell\in P_i} \mathbb{E}[c_\ell^i(\mathbf{F})]$ for any $i\in I$, then \eqref{SN-UE-poly-f} is equivalent to
\begin{equation*}
\begin{cases}
f_k^i( \mathbb{E}[c_k^i(\mathbf{F})]-\pi_i)=0,\\
f_k^i\geq 0,\\
\end{cases}\quad\forall\ k\in P_i,\ \forall\ i\in I.
\end{equation*}

Let $\mathbf{p}^\ast$, $\mathbf{F}^\ast$ and $\mathbf{f}^\ast$ be the vectors of strategies and the corresponding
path flows, mean path flows at the UE-SD, respectively. Then
\[
\underset{i\in I}\sum\underset{k\in P_i}\sum (f_k^i)^\ast( \mathbb{E}[c_k^i(\mathbf{F}^\ast)]-\pi_i)=0.
\]
For any feasible $\mathbf{f}=(f_k^i\geq 0: k\in P_i, i\in I)$, we also have
\[
\underset{i\in I}\sum\underset{k\in P_i}\sum f_k^i( \mathbb{E}[c_k^i(\mathbf{\mathbf{F}}^\ast)]-\pi_i)\geq 0.
\]
Thus
\begin{equation}\label{Prop18-3}
\underset{i\in I}\sum\underset{k\in P_i}\sum (f_k^i)^{\ast}(\mathbb{E}[c_k^i(\mathbf{F}^\ast)]-\pi_i)\leq \underset{i\in I}\sum\underset{k\in P_i}\sum f_k^i(\mathbb{E}[c_k^i(\mathbf{F}^\ast)]-\pi_i).
\end{equation}
From the feasibility condition \eqref{eqn:feasible_flow} we have $\sum_{k\in P_i} f_k^i=\sum_{k\in P_i} (f_k^i)^\ast=d_i$ for every $i\in I$. Hence
\[
\underset{i\in I}\sum\underset{k\in P_i}\sum (f_k^i)^{\ast} \pi_i=\underset{i\in I}\sum\underset{k\in P_i}\sum f_k^i\pi_i,
\]
which together with \eqref{Prop18-3} implies \eqref{eqn:UE-SD-VI}:
\[
\underset{i\in I}\sum\underset{k\in P_i}\sum (f_k^i)^{\ast}\mathbb{E}[c_k^i(\mathbf{F}^\ast)]\leq \underset{i\in I}\sum\underset{k\in P_i}\sum f_k^i \mathbb{E}[c_k^i(\mathbf{F}^\ast)].
\]

On the other hand, observe that as the first order optimality condition, the solution of VI problem
\eqref{eqn:UE-SD-VI} also solves the following LP problem:
\begin{equation*}
\begin{array}{cll}
\min &  \mathbf{f}^T \mathbb{E}[\mathbf{c}(\mathbf{F}^\ast)]&\\
\text{s.t.}& \displaystyle\sum_{k\in P_i} f_k^i=d_i,  & i\in I,\\
&f_k^i\geq 0, & k\in P_i, \,i\in I,
\end{array}
\end{equation*}
the duality of which is
\[
\begin{array}{cll}
\displaystyle \max &  \mathbf{\lambda}^T \mathbf{d} &\\
\text{s.t.}& \displaystyle \lambda_i \le \mathbb{E}[c_k^i(\mathbf{F}^\ast)],  &k\in P_i, \, i\in I.
\end{array}
\]
Therefore, we have the following complementary slackness conditions:
\[
(\mathbb{E}[c_k^i(\mathbf{F}^\ast)]-\lambda_i)f_k^i =0, \ k\in P_i,\, i\in I,
\]
which imply \eqref{eqn:def_UE-SD}.
\end{proof}

An equivalence between the UE-SD condition and a minimization problem can also be established if the link cost functions are linear, which is stated as in the following proposition.
\begin{prop}\label{pro:VI-min}
Given a transportation game $(G, \mathbf{D}, \mathbf{c})$ with cost functions $\mathbf{c}$ linear, let $\mathbf{p^\ast}\in\Omega$ be a vector of routing choice probabilities. Then $\mathbf{p^\ast}$ is a UE-SD if and only if it solves the following minimization problem
\begin{equation}\label{eqn:UE-SD-MIN-P}
\min_{\mathbf{p} \in \Omega} Z(\mathbf{p})\equiv\sum_{e\in E}\int_0^{v_e} c_e(x) d x,
\end{equation}
where $v_e =\sum_{i\in I}\sum_{k\in P_i}\delta_{k,e}^i p_k^i d_i$.
\end{prop}
\begin{proof}
We prove this proposition by verifying the equivalence between VI problem \eqref{eqn:UE-SD-VI} and
minimization problem \eqref{eqn:UE-SD-MIN-P}. Note that, since the link cost function $c_e(x)$ is continuously differentiable and non-decreasing, function $\int_0^{v_e} c_e(x) d x$ is convex (with respect to $v_e$) for any $e\in E$, which together with the fact that convexity is invariant under affine maps implies that the objective function
$Z(\mathbf{p})$ in \eqref{eqn:UE-SD-MIN-P} is convex, which together with the fact that the feasible region
$\Omega$ is convex and compact implies in turn that is a convex optimization problem.
Therefore, it is necessary and sufficient for $\mathbf{p}^*$ to satisfy the first order optimality condition of \eqref{eqn:UE-SD-MIN-P} \Citep[Proposition 2.1.2]{Bertsekas1999}:
\begin{equation}\label{eqn:1st-order-linear1}
(\mathbf{p}-\mathbf{p}^{\ast})^T \nabla Z(\mathbf{p^*}) \geq 0.
\end{equation}
Since
\[
\frac{\partial Z(\mathbf{p})}{\partial p_k^i}=\sum_{e\in E} c_e(v_e) \frac{\partial v_e}{\partial p_k^i}
=\sum_{e\in E} c_e(v_e) ( \delta_{k,e}^i d_i )=\ c_k^i(f)d_i,
\]
which together with $\left(f_k^i\right)^{\ast}=
\left(p_k^i\right)^{\ast}d_i$ due to \eqref{eqn:flow_distribution} implies that condition \eqref{eqn:1st-order-linear1} is equivalent to
\begin{equation*}
\quad (\mathbf{f}-\mathbf{f}^{\ast})^T \mathbf{c}(\mathbf{f}^{\ast})\geq 0,
\end{equation*}
which is equivalent to \eqref{eqn:UE-SD-VI} when the link cost functions are linear.
\end{proof}

In general, when link cost functions are nonlinear, we rewrite the UE-SD equivalent condition
\eqref{eqn:UE-SD-VI} in the following form by substituting $f_k^i=p_k^id_i$ and
$(f_k^i)^\ast=(p_k^i)^\ast d_i$:
\begin{equation}\label{eqn:VI_substitute}
(\mathbf{p}-\mathbf{p}^\ast)^T \mathbf{S}(\mathbf{p}^\ast)\geq 0, \ \mathbf{p} \in \Omega
\end{equation}
where $\mathbf{S}(\mathbf{p}^\ast)$ is a vector with the same dimension as $\mathbb{E}[\mathbf{c}(\mathbf{F})]$,
obtained by replacing element $\mathbb{E}[c_k^i(\mathbf{F})]$ in vector $\mathbb{E}[\mathbf{c}(\mathbf{F})]$ with
$\mathbb{E}[c_k^i(\mathbf{F})]d_i$ for every $k\in P_i,\ i\in I$.

Proposition~\ref{pro:VI-min} establishes that the VI condition for a UE-SD is just a restatement of the first
order necessary and sufficient condition of a convex minimization problem, if the cost functions $\mathbf{c}$
are linear. We use the more general VI condition
\eqref{eqn:VI_substitute} to establish the following general conditions for a UE-SD to exist and to be unique.

\begin{prop}[Existence and uniqueness of the UE-SD]
Let $(G,\mathbf{D}, \mathbf{c})$ be a transportation game. Then: (a) the game admits at least one UE-SD if the
link cost functions are continuous. Furthermore, (b) the UE-SD is unique if $\nabla \mathbf{S}(\mathbf{p})$ is positive definite.
\end{prop}
\begin{proof}
(a) The existence of a solution $\mathbf{p}^*\in \Omega$ to \eqref{eqn:VI_substitute} is implied by the
continuity of $\mathbf{S}(\mathbf{p})$ and compactness of $\Omega$. (b) The uniqueness is implied by the
positive definiteness of the Jacobian matrix of $\mathbf{S}(\mathbf{p})$
\citep[see][Proposition 1.5 and Theorem 1.8]{Nagurney1998}.
\end{proof}

NB: When the link cost functions are affine and strictly monotone, then $\nabla\mathbf{S}(\mathbf{p})$
is positive definite.

\subsection{System optimum under stochastic demand (SO-SD)}

At a system optimum (SO-SD), traffic is coordinated by a central authority according to mixed strategies.
It should be noted in the case of coordination that traffic is assigned according to route choice probabilities
rather than by traffic proportions. This is due to the fact that demand is cumulative over the time period,
while traffic allocation needs to be made once a traffic flow arrives the route entrance. The central
authority has to implement traffic coordination without full knowledge of the actual demand. The objective of
for the coordinator is to minimize the expectation of the total travel cost at an SO-SD. This gives rise to
our following definition.
\begin{defn}[SO-SD condition]
Given a transportation game $(G, \mathbf{D}, \mathbf{c})$ with stochastic demands, a vector
$\mathbf{p}\in \Omega$ of routing choice probabilities is said to be an SO-SD strategy if it
solves the following minimization problem
\begin{equation}\label{eqn:SO-SD}
\min_{\mathbf{p}\in\Omega}\, T(\mathbf{p}) \equiv \,
   \mathbb{E}\hspace{-2pt}\left[\underset{e\in E}\sum c_e(V_e)V_e \right],
\end{equation}
where $V_e$ is a function of $\mathbf{p}$ given by \eqref{eqn:Ve} and \eqref{eqn:flow_distribution}.
\end{defn}

\section{Price of anarchy}

In this section we investigate the price of anarchy (PoA) to be defined below based on the model presented
in the preceding section with the expected total cost $T(\cdot)$ defined in the network by \eqref{eqn:SO-SD}
as the social (system)
objective function. Given an instance $(G, \mathbf{D}, \mathbf{c})$ of the transportation game with
stochastic demands, the corresponding PoA is defined as the worst-case ratio between expect total costs
at UE-SD and at SO-SD:
\[
\mbox{PoA}(G,\mathbf{D},\mathbf{c}):=\max\left\{\frac{T(\mathbf{p})}{T(\mathbf{q})}:
\mathbf{p},\mathbf{q}\in\Omega, \textrm{$\mathbf{p}$ is UE-SD and $\mathbf{q}$ is SO-SD}\right\}.
\]
Let $\mathcal{I}$ be the set of all instances of the transportation game with stochastic demands,
then the PoA of the problem of transportation game with stochastic demands is defined as
\[
\mbox{PoA}(\mathcal{I}):=\max_{(G, \mathbf{D}, \mathbf{c})\in\mathcal{I}} \mbox{PoA}(G, \mathbf{D}, \mathbf{c}).
\]
Note that even for deterministic demands (i.e., $\mathbf{D}$ is particularly deterministic), the PoA is
already unbounded if the link cost functions $\mathbf{c}$ are unrestricted \Citep{Roughgarden2002}.
In this study, we will establish upper bounds of the PoA for a fixed set $\mathcal{C}$ of link cost functions,
namely, the set of affine cost functions and that of polynomial cost functions.

\subsection{Affine cost functions}
\label{sec:affine}

Let us first consider affine link cost functions, i.e.,
\begin{equation}\label{eqn:affine_cost_function}
    c_e(x)=a_ex +b_e, \mbox{ where } a_e, b_e\geq 0, \ e\in E.
\end{equation}
Given demand distributions $\mathbf{D}$, the means $\{v_e\}$ and variances $\{\sigma_e^2\}$ of the link
flows can be derived from the link-path conservation equation \eqref{eqn:Ve} as follows:
\begin{equation}\label{eqn:affine_implied_relation}
\left\{
\begin{array}{cl}
  v_e & =\sum_{i\in I}\sum_{k\in P_i}\delta_{k,e}^ip_k^id_i, \\
  \sigma_e^2 & = \mbox{Var}\left[\sum_{i\in I,\,k\in P_i}\delta_{k,e}^i p_k^i D_i \right] \\
             & =\hbox{Var}\left[\sum_{i\in I} \delta_e^i p_e^i D_i\right]
               =\sum_{i\in I}\delta_e^i(p_e^i)^2 \sigma_i^2,
\end{array}
\right.
\end{equation}
where
\begin{equation}\label{eqn:def_p-e-i}
    p_e^i=\sum_{k\in P_i}\delta_{k,e}^i p_k^i
\end{equation}
and the last equality is obtained from
the independence of the demands of different O-D pairs. Let $\Omega_0 \subseteq \Omega$ be the set
of those $\mathbf{p}\in\Omega$ that additionally satisfy \eqref{eqn:affine_implied_relation} and
\eqref{eqn:def_p-e-i}. According to Proposition~\ref{pro:VI-min},
the (unique) UE-SD with affine cost functions is the same as the optimal solution to the following problem:
\begin{equation}\label{eqn:UE-SD-linear}
\underset{\mathbf{p}\in\Omega_0}\min \,\sum_{e\in E} \left(\frac{a_e}{2}v_e^2+b_ev_e\right).
\end{equation}
On the other hand, the (unique) SO-SD strategy solves the following problem:
\begin{equation}\label{eqn:SO-SD-linear}
\displaystyle \underset{\mathbf{p}\in\Omega_0} \min\, \underset{e\in E}
 \sum \left({a_e}(v_e^2+\sigma_e^2)+b_ev_e\right).
\end{equation}
Before proceeding, let us consider the following problem with $\alpha>0$ constant:
\begin{equation}\label{eqn:Hv}
\underset{\mathbf{p}\in\Omega_0}\min \, H(\mathbf{v})=\underset{e\in E}
  \sum \left(\alpha\cdot a_e v_e^2+b_e v_e \right).
\end{equation}
\begin{lem}\label{lem:Hv}
Given any instance $(G, \mathbf{D}, \mathbf{c})$ of the transportation game with affine link cost functions
$\mathbf{c}$. Routing choice strategy $\mathbf{p}\in\Omega$ is a UE-SD if and only if it is an optimal
solution to problem \eqref{eqn:Hv} for instance $(G,\bar{\mathbf{D}},c)$ with $\bar{\mathbf{D}}=\mathbf{D}
/(2\alpha)$.
\end{lem}
\begin{proof}
First note that both problem \eqref{eqn:UE-SD-linear} for $(G, \mathbf{D}, \mathbf{c})$ and problem
\eqref{eqn:Hv} for instance $(G,\mathbf{D}/(2\alpha),\mathbf{c})$ have the same feasible region $\Omega_0$ as
$\mathbf{\mathbf{p}}$ only plays a role of linking $\{v_e\}$ and $\{\sigma_e^2\}$ with $\{d_i\}$ and
$\{\sigma_i^2\}$ in \eqref{eqn:affine_implied_relation} (i.e., $\Omega_0 = \Omega_0(\mathbf{D})=\Omega_0(\bar{\mathbf{D}})$).
On the other hand, since
$\bar{v}_e={v_e}/({2\alpha})$, we have
\[
H(\mathbf{\bar{v}})= \frac{1}{2\alpha}\underset{e\in E}\sum \left(\frac{a_e}{2}v_e^2+b_ev_e\right).
\]
In other words, the objectives of the two problems differ only by a constant $1/(2\alpha)$. Therefore,
they have the same optimal solution.
\end{proof}

Lemma~\ref{lem:Hv} provides us with a parametric ($\alpha$) function $H(\cdot)$ to quantify a UE-SD.
Let us start with a lower bound.

\begin{lem}\label{lem:H_lb}
Let $(G, \mathbf{D}, \mathbf{c})$ be a transportation game with stochastic demands and affine link
costs \eqref{eqn:affine_cost_function}. Let $\mathbf{p}^*$ be the optimal solution to convex program
\eqref{eqn:Hv} for $(G, \mathbf{D}, \mathbf{c})$ and $\mathbf{v}^{\ast}$ be the corresponding vector of the
mean link flows. Then for any $\mathbf{p}\in \Omega$, the corresponding vector $\mathbf{v}_\beta$ of
the mean link flows for $(G,\beta\mathbf{D},\mathbf{c})$ for some fixed $\beta>1$ satisfies the following inequality:
\begin{equation} \label{eqn:H_lb}
H(\mathbf{v}_\beta)\geq H(\mathbf{v}^\ast) + (\beta-1)\underset{e\in E}\sum v_e^\ast h_e'(v_e^\ast),
\end{equation}
where $h_e(x)=\alpha \cdot a_ex^2+b_e x$ and $h_e'(\cdot)$ is the derivative of $h_e(\cdot)$ for $e\in E$.
\end{lem}
\begin{proof}
Since $h_e(\cdot)$ is convex, we have a lower bound on the linear approximation at the point $v_e^\ast$
\[
h_e(v_e)\geq h_e(v_e^{\ast})+(v_e-v_e^{\ast}) h_e'(v_e^{\ast}),\quad  \forall\ e\in E,
\]
which leads to
\begin{equation}\label{L32}
H(\mathbf{v})\geq H(v^\ast)+ \underset{e\in E}\sum (v_e-v_e^{\ast}) h_e'(v_e^{\ast}).
\end{equation}
Since $\mathbf{p}^{\ast}$ is optimal for convex program \eqref{eqn:Hv}, the first order optimality condition
gives
\begin{equation}\label{L33}
(\mathbf{p}-\mathbf{p}^\ast)^T \nabla_{\mathbf{p}^\ast} H(\mathbf{v}^\ast)\geq 0.
\end{equation}
On the other hand, since $\mathbf{v}^\ast$ and $\mathbf{v}_\beta/{\beta}$ are the corresponding vectors of the
mean link flows for (the same) game $(G, \mathbf{D}, \mathbf{c})$,
respectively to strategies $\mathbf{p}^\ast$ and $\mathbf{p}$, from relations
\eqref{eqn:affine_implied_relation} we obtain, for any $e\in E$,
$v_e^\ast=\sum_{i\in I\,k\in P_i}\delta_{k,e}^i(p_k^i)^\ast d_i$ and
${v_e}/{\beta}=\sum_{i\in I\,k\in P_i}\delta_{k,e}^i p_k^i d_i$.
Hence
\begin{align*}
(\mathbf{p}-\mathbf{p}^\ast)^T \nabla_{\mathbf{p}^\ast} H(\mathbf{v}^\ast)&=\underset{i\in I}\sum \underset{k\in P_i}\sum \left(p_k^i-(p_k^i)^\ast \right)\frac{\partial H(\mathbf{v}^\ast)}{\partial (p_k^i)^\ast} \nonumber\\
&=\underset{i\in I}\sum \underset{k\in P_i}\sum \left[(p_k^i-(p_k^i)^\ast)\underset{e\in E}\sum\frac{\partial H(\mathbf{v}^\ast)}{\partial v_e^\ast}\cdot \frac{\partial v_e^\ast}{\partial(p_k^i)^\ast}\right] \nonumber\\
&=\underset{e\in E}\sum \left[\frac{\partial H(\mathbf{v}^\ast)}{\partial v_e^\ast}\cdot \underset{i\in I}\sum \underset{k\in P_i}\sum (p_k^i-(p_k^i)^\ast) \delta_{k,e}^id_i\right] \nonumber\\
&=\underset{e\in E}\sum \frac{\partial H(\mathbf{v}^\ast)}{\partial v_e^\ast}\cdot \left(\frac{v_e}{\beta}-v_e^\ast\right),
\end{align*}
which together with \eqref{L33} implies
\[
\underset{e\in E}\sum h_e'(v_e^{\ast})v_e\geq \beta \underset{e\in E}\sum h_e'(v_e^{\ast})v_e^{\ast},
\]
which together with \eqref{L32} implies \eqref{eqn:H_lb}.
\end{proof}

The following lemma establishes two functions of the mean link flows to bound the expected total cost
of the entire network.

\begin{lem}\label{lem:TC_bounds}
Let $(G, \mathbf{D}, \mathbf{c})$ be a transportation game with stochastic demands and affine link
costs \eqref{eqn:affine_cost_function}. Let $\mathbf{p}\in\Omega$ be any feasible routing choice
strategy. Then the expected total cost $T(\mathbf{p})$ is bounded from both below and above as follows (see Section~\ref{sec:with_notation} for notation):
\[
\sum_{e\in E}\left(\left(1+\frac{{\underline{\theta}}^2}{n}\right)a_ev_e^2+b_e v_e \right) \leq T(\mathbf{p})\leq\underset{e\in E}\sum \left((1+{\bar{\theta}}^2)a_ev_e^2+b_e v_e \right).
\]
\end{lem}
\begin{proof}
Noticing that $\theta_i=\sigma_i/d_i$, from \eqref{eqn:affine_implied_relation} we bound $\sigma_e^2$ from
above:
\begin{eqnarray}
\sigma_e^2 &= &\underset{i\in I}\sum\delta_e^i(p_e^i)^2\theta_i^2 d_i^2
\leq \left(\underset{i\in I}\max\{\theta_i\}\right)^2 \underset{i\in I}\sum\delta_e^i(p_e^i d_i)^2 \nonumber\\
&\leq & \overline{\theta}^2\left(\sum_{i\in I,\,k\in P_i}\delta_{k,e}^ip_k^i d_i\right)^2
= \overline{\theta}^2 v_e^2, \label{eqn:sigma_ub}
\end{eqnarray}
and bound $\sigma_e^2$ from below:
\begin{eqnarray}
\sigma_e^2 &\geq & \left(\min_{i\in I}\{\theta_i\}\right)^2 \underset{i\in I}\sum\delta_e^i(p_e^i d_i)^2
\geq\ \frac{\underline{\theta}^2}{n_e} \left(\underset{i\in I}\sum \delta_e^ip_e^id_i\right)^2
  \nonumber \\
&=& \frac{\underline{\theta}^2}{n_e}\left(\sum_{i\in I,\,k\in P_i}\delta_{k,e}^i p_k^id_i
   \right)^2
\geq  \frac{\underline{\theta}^2}{n} v_e^2, \label{eqn:sigma_lb}
\end{eqnarray}
where the second inequality in \eqref{eqn:sigma_lb} follows from Cauchy-Schwarz inequality.
Combination of the upper and lower bounds \eqref{eqn:sigma_ub} and \eqref{eqn:sigma_lb}
leads directly to the desired inequalities in the lemma.
\end{proof}

Now we are ready to present our first main result.

\begin{thm}\label{thm:PoA1}
Let $(G, \mathbf{D}, \mathbf{c})$ be a transportation game with stochastic demands and affine link
costs \eqref{eqn:affine_cost_function}. Then
\[
\mbox{PoA}(G,\mathbf{D},\mathbf{c})\leq \frac{4(1+{\bar{\theta}}^2)(n+{\underline{\theta}}^2)}{(3n+4{\underline{\theta}}^2)}.
\]
\end{thm}
\begin{proof}
Fix $\alpha=1+\underline{\theta}^2 /n$ in the definition of $H(\mathbf{v})$ in \eqref{eqn:Hv}. Let
$\mathbf{p}\in\Omega$ be the UE-SD strategy for $(G,\mathbf{D},\mathbf{c})$ and $\mathbf{v}$ be the
vector of the corresponding mean link flows. According to Lemma~\ref{lem:Hv}, $\mathbf{p}$ is also
the optimal solution to convex program \eqref{eqn:Hv} for $(G, \mathbf{D}/(2\alpha),\mathbf{c})$, for which
the vector of the corresponding mean link flows is $\mathbf{v}/(2\alpha)$.

Let $\mathbf{p}^*\in\Omega$ be the SO-SD strategy for $(G, \mathbf{D}, \mathbf{c})$ and $\mathbf{v}^\ast$ be the
vector of the corresponding mean link flows. By applying Lemma~\ref{lem:H_lb} with $\beta=2\alpha$,
we have $\mathbf{v}^* = \mathbf{v}_\beta$ and
\begin{eqnarray}
  H(\mathbf{v}^\ast) &\ge& H\hspace{-2pt}\left(\frac{\mathbf{v}}{2\alpha}\right)
          +\left(2\alpha-1\right)\sum_{e\in E}\left(\frac{v_e}{2\alpha}\right)
          h_e'\hspace{-2pt}\left(\frac{v_e}{2\alpha}\right)    \nonumber       \\
   &=&  \sum_{e\in E} \left(\alpha a_e\frac{v_e^2}{4\alpha^2}+b_e\frac{v_e}{2\alpha}\right)
          + \frac{2\alpha-1}{2\alpha}\underset{e\in E}\sum(a_ev_e^2+b_ev_e)   \nonumber \\
   &\ge&  \frac{1}{4\alpha^2}\sum_{e\in E}\left(\alpha a_e v_e^2+b_ev_e \right)
          + \frac{2\alpha-1}{2\alpha^2}\underset{e\in E}\sum \left(\alpha a_e v_e^2+b_ev_e \right) \nonumber \\
   &=& \frac{4\alpha-1}{4\alpha^2}\,H(\mathbf{v}). \label{eqn:proof_thm_PoA}
\end{eqnarray}
Now applying Lemma~\ref{lem:TC_bounds} and noticing $\alpha \le 1+{\bar{\theta}}^2$, we obtain
\begin{eqnarray*}
\mbox{PoA}(G,\mathbf{D},\mathbf{c}) &=&  \frac{T(\mathbf{p})}{T(\mathbf{p}^*)}
          \le \frac{\sum_{e\in E} \left((1+{\bar{\theta}}^2) a_ev_e^2+b_ev_e\right)}{\sum_{e\in E}
          \left(\alpha a_e(v_e^\ast)^2+b_ev_e^\ast\right)}   \\
   &\le&  \frac{1+{\bar{\theta}}^2}{\alpha}\,\frac{\sum_{e\in E}(\alpha a_ev_e^2+b_ev_e)}
          {\sum_{e\in E}(\alpha a_e(v_e^\ast)^2+b_ev_e^\ast)}
          = \frac{1+{\bar{\theta}}^2}{\alpha}\,\frac{H(\mathbf{v})}{H(\mathbf{v}^\ast)} \\
   &\le&  \frac{1+{\bar{\theta}}^2}{\alpha}\,\frac{4\alpha^2}{4\alpha-1}
          = \frac{4(1+{\bar{\theta}}^2)(n+{\underline{\theta}}^2)}{(3n+4{\underline{\theta}}^2)},
\end{eqnarray*}
where the last inequality follows from inequality \eqref{eqn:proof_thm_PoA}.
\end{proof}

\noindent\textbf{Remarks.} To some extent the upper bound on PoA in Theorem \ref{thm:PoA1} is tight as
demonstrated in the following example, where all coefficients of variation of demand distributions are equal
and each link is in at most one path of each O-D pair.

\begin{eg}\label{eg:affine}
 Consider a single commodity network as shown in Figure~\ref{fig:pigou_affine}, in which the stochastic demand
 $X$ has a mean of $d$ and variance of $\sigma^2$. Denote $\theta= \sigma /d$ as the coefficient of variation.
 The link cost on the upper link is a constant of the mean demand and that on the lower link is the traffic amount.

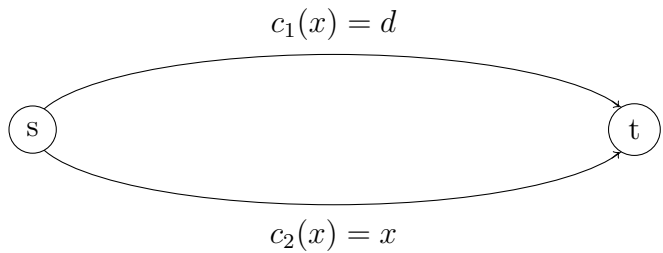
\begin{figure}[H]
\centering
\begin{tikzpicture}
\draw (0,0) arc (0:360:4 and 1);
\draw [->] (-0.3,0.38)--(-0.2,0.3);
\draw [->] (-0.3,-0.38)--(-0.2,-0.3);
\node at (0,0) [circle,draw,fill=white] {t};
\node at(-8,0)[circle,draw,fill=white] {s};
\node at(-4, -1.4){$c_2(x)=x$};
\node at(-4, 1.4){$c_1(x)=d$};
\end{tikzpicture}
\caption{Two-link network}\label{fig:pigou_affine}
\end{figure}

Since the expected travel cost on the lower link is always no more than that on the upper link, strategy
$\mathbf{p}^T=(0,1)^T$ to choose the lower link with probability $1$ is a UE-SD. The expected total cost
$T(\mathbf{p})=\mathbb{E}[X^2]=(1+\theta^2)d^2$.
Let $\mathbf{p}^\ast=(p_1^\ast,p_2^\ast)^T$ be the SO-SD strategy. Solving \eqref{eqn:SO-SD-linear}
we get
\begin{displaymath}
\begin{cases}
p_1^\ast=1-1/\left(2(1+\theta^2)\right),\\
p_2^\ast=1/\left(2(1+\theta^2)\right),
\end{cases}
\end{displaymath}
from which we obtain $T(\mathbf{p}^*)={(4\theta^2+3)d^2}/{(4(1+\theta^2))}$. Thus the value of PoA is
\[
\frac{T(\mathbf{p})}{T(\mathbf{p}^*)}=\frac{4(1+\theta^2)^2}{3+4\theta^2},
\]
\end{eg}
which matches the upper bound in Theorem~\ref{thm:PoA1} with $n=1$ and $\theta=\underline{\theta}
=\overline{\theta}$.

\subsection{Polynomial cost functions and positive-valued distributions}
\label{sec:poly}

Now let us consider polynomial link cost functions of
\begin{equation}\label{eqn:poly_cost_function}
c_e(x) =\sum_{j=0}^{m}b_{ej}x^j,\ b_{ej}\geq 0,\ j=0,1,\ldots, m;\ e\in E.
\end{equation}

Assume the traffic demand follows a positive-valued distribution and has finite $m$-th moment.
Define the following parameters
\begin{equation}\label{eqn:theta_m}
\theta_i^{(m)}=\frac{\mathbb{E}[D_i^m]}{d_i^m}, \quad \forall\ i\in I,
\end{equation}
to denote the ratio of the $m$-th moment of the demand to the mean demand to the power of $m$.
Denote $\overline{\theta}^{(m)}=\max_{i\in I}\left\{\theta_i^{(m)}\right\}$. The following lemma
establishes two inequalities between the $m$-th moment of $V_e$ and $v_e^m$ for every $e\in E$.

\begin{lem}\label{lem:moment_Ve}
For any transportation game $(G, \mathbf{D}, \mathbf{c})$ in which $\mathbf{D}$ is a vector of
positive-valued distributions, we have
\[
v_e^m \leq \mathbb{E}[V_e^m]\leq \overline{\theta}^{(m)}\cdot v_e^m,\quad e\in E.
\]
\end{lem}
\begin{proof}
The first inequality in the lemma follows from Jensen's Inequality. Next we prove the second inequality.
For any two nonnegative integers $s,t \in \mathbb{Z}_+$ satisfing $s+t=m$,
\[
\mathbb{E}[D_i^m]=\mathbb{E}[D_i^sD_i^t]=\mathbb{E}[D_i^s]\mathbb{E}[D_i^t]+\hbox{Cov}(D_i^s,D_i^t), \quad  i \in I.
\]
Since $D_i$ is a positive random variable, $\hbox{Cov}(D_i^s,D_i^t)\geq0$ (see, e.g., \Citet{Schmidt2003}). Thus
\[
\mathbb{E}[D_i^m]\geq \mathbb{E}[D_i^s]\mathbb{E}[D_i^t], \quad  i\in I,
\]
which leads to
\[ \frac{\mathbb{E}\left[D_i^m\right]}{\mathbb{E}[D_i]^m}
\geq\frac{\mathbb{E}[D_i^s]}{\mathbb{E}[D_i]^s}\frac{\mathbb{E}[D_i^t]}{\mathbb{E}[D_i]^t},\quad i\in I.
\]
With definition \eqref{eqn:theta_m} we have
\begin{equation}\label{eqn:theta_factorization}
\theta_i^{(m)}\geq \theta_i^{(s)} \cdot
\theta_i^{(t)}, \quad i\in I.
\end{equation}

For any $e\in E$, denote $S_m=\{\mathbf{s}=(s_i:i\in I)
\in \mathbb{Z}_+^{|I|}: \sum_{i\in I} s_i =m\}$. Then with \eqref{eqn:def_p-e-i} the $m$-th moment
of the link flow can be written as
\begin{eqnarray*}
\mathbb{E}[V_e^m] &=& \mathbb{E}\left[\left(\underset{i\in I}\sum p_e^i D_i\right)^m\right]
    = \mathbb{E}\left[\sum_{\mathbf{s}\in S_m}\frac{m!}{\prod_{i\in I} s_i!}
    \prod_{i\in I}(p_e^iD_i)^{s_i}\right]   \\
&=& \sum_{\mathbf{s}\in S_m}\frac{m!}{\prod_{i\in I} s_i!}
    \prod_{i\in I}\mathbb{E}\left[(p_e^iD_i)^{s_i}\right]
    = \sum_{\mathbf{s}\in S_m}\frac{m!}{\prod_{i\in I} s_i!}
    \prod_{i\in I}\theta_i^{(s_i)}(p_e^id_i)^{s_i},
\end{eqnarray*}
where the third equality is due to the independence among demands of different O-D pairs, which together
with (due to \eqref{eqn:theta_factorization})
\[
\max_{i\in I}\left\{\theta_i^{(m)}\right\}\geq\prod_{i\in I}\theta_i^{(s_i)}, \quad \mathbf{s}\in S_m
\]
implies that
\begin{eqnarray*}
\mathbb{E}[V_e^m] &\leq & \max_{i\in I}\left\{\theta_i^{(m)}\right\}
      \sum_{\mathbf{s}\in S_m}\frac{m!}{\prod_{i\in I} s_i!}\prod_{i\in I}(p_e^id_i)^{s_i} \\
&=& \overline{\theta}_i^{(m)}\left(\sum_{i\in I} p_e^id_i\right)^m
    = \overline{\theta}_i^{(m)}v_e^m,
\end{eqnarray*}
completing our proof of the lemma.
\end{proof}

As shown in \citep{Roughgarden2005-book}, the PoA in any deterministic model is bounded by
the \textit{anarchy value}, which depends only on the class of the (link) cost functions.
We extend this approach to stochastic models.

\begin{defn}\label{def:aux_cost-functions}
Let $c(x)=\sum_{j=0}^m b_jx^j$ with all $b_j\ge0$ and $\mathbf{D}$ be a
vector of positive random variables. Define:
\[
\left\{
  \begin{array}{l}
  \overline{t}(x)=\sum_{j=0}^{m} b_{j}\overline{\theta}^{(j)}x^j,\\
  \underline{t}(x)=\sum_{j=0}^{m} b_{j}x^j;
  \end{array}
\right.
\mbox{ and }
\left\{
  \begin{array}{l}
  \overline{s}(x)=\sum_{j=0}^{m} b_{j}\overline{\theta}^{(j+1)}x^{j+1}, \\
  \underline{s}(x)=\sum_{j=0}^{m} b_{j}x^{j+1}.
  \end{array}
\right.
\]
The derivative of $\underline{s}(x)$ is $\underline{s}'(x)=\sum_{j=0}^{m}(j+1) b_{j}x^{j}$.
Since $\overline{t}(x)$ and $\overline{s}'(x)$ are both functions with domain $x\in (0,+\infty)$ and range $(b_0,+\infty)$, there exists a value $\lambda(x)>0$, for any $x>0$, such that $\underline{s}'(\lambda(x) x)=\overline{t}(x)$.
Define
\begin{eqnarray*}
\mu(x) &=& \underline{s}(\lambda(x)x)/ \overline{s}(x),  \\
\phi(x) &=& x\underline{t}(x)/\overline{s}(x), \\
\eta(x) &=& x\overline{t}(x)/\overline{s}(x),
\end{eqnarray*}
with the understanding $0/0=1$.
\end{defn}
\begin{defn}
Let $c(x)=\sum_{j=0}^m b_jx^j$ with all $b_j\ge0$ and $\mathbf{D}$ be a
vector of positive random variables, such that
\begin{equation}\label{eqn:poly_function_restriction}
\underset{x>0}\min\  \{\mu(x)+\phi(x)-\eta(x)\lambda(x)\}>0.
\end{equation}
Define
\[
\gamma(c,\mathbf{D})= \underset{x>0}\sup\ [\mu(x)+\phi(x)-\eta(x)\lambda(x)]^{-1},
\]
and
\[
\gamma(\mathcal{C},\mathbf{D})=\underset{c\in \mathcal{C}}\sup\ \gamma(c, \mathbf{D}),
\]
where $\mathcal{C}$ is a subset of positive polynomial cost functions that satisfy \eqref{eqn:poly_function_restriction}.
\end{defn}

For link cost functions \eqref{eqn:poly_cost_function}, we define as in Definition~\ref{def:aux_cost-functions}
the corresponding functions $\overline{t}_e(v_e)$, $\underline{t}_e(v_e)$, $\overline{s}_e(v_e)$,
and $\underline{s}_e(v_e)$. Then it follows from Lemma~\ref{lem:moment_Ve} that
\begin{align}
&\underline{s}_e(v_e)\leq \mathbb{E}[c_e(V_e)V_e]\leq \overline{s}_e(v_e), &\forall\ e\in E,
   \label{eqn:expected_link_cost_bound0}\\
&\underline{t}_e(v_e)\leq \mathbb{E}[c_e(V_e)]\leq \overline{t}_e(v_e), &\forall\ e\in E.
    \label{eqn:expected_link_cost_bound}
\end{align}

Next lemma provides a weaker version of the UE-SD condition with only the mean link flows.

\begin{lem} \label{lem:weaker-UE-SD}
Given a transportation game $(G, \mathbf{D}, \mathbf{c})$ with polynomial (link) cost functions \eqref{eqn:poly_cost_function} and positive demand distributions. Let $\mathbf{p}^{\ast}$ be a UE-SD. Then for an arbitrary strategy $\mathbf{p}$, we have
\[
\underset{e\in E}\sum v_e\, \overline{t}_e(v_e^{\ast})\geq \underset{e\in E}\sum v_e^{\ast}\, \underline{t}_e(v_e^{\ast}),
\]
where $\mathbf{v}^{\ast}$ and $\mathbf{v}$ are the mean link flows corresponding to $\mathbf{p}^{\ast}$ and $\mathbf{p}$, respectively.
\end{lem}
\begin{proof}
The lemma can be proved by applying \eqref{eqn:expected_link_cost_bound} into VI problem \eqref{eqn:UE-SD-VI}.
\end{proof}

As in Definition~\ref{def:aux_cost-functions}, we obtain $\lambda_e(\cdot)>0$ for each $e\in E$.
Then we have a lower bound of the expected total cost at the SO-SD in the following lemma.

\begin{lem}\label{lem:Se_lb}
Given a transportation game $(G, \mathbf{D}, \mathbf{c})$ with polynomial link cost functions \eqref{eqn:poly_cost_function} and positive demand distributions. Let $\mathbf{v}$ and
$\mathbf{v}^{\ast}\in \mathbb{R}^{|E|}$ be the vectors of mean flows at the UE-SD and the
SO-SD respectively. Then,
\[
\sum_{e\in E}\underline{s}_e(v_e^\ast)\geq \underset{e\in E}\sum (\underline{s}_e(\lambda_e(v_e)v_e)+ (v_e^{\ast}-\lambda_e(v_e)v_e)\overline{t}_e(v_e)).
\]
\end {lem}

\begin{proof}
Since $\underline{s}'_e(\lambda_e(v_e) v_e)=\overline{t}_e(v_e)$ for any $e\in E$, with convexity of
$\underline{s}_e(\cdot)$ we obtain
\begin{align*}
\underline{s}_e(v_e^\ast)\geq\ & \underline{s}_e(\lambda_e(v_e)v_e)+(v_e^{\ast}-\lambda_e(v_e)v_e)\underline{s}_e'(\lambda_e(v_e)v_e)\\
=\ & \underline{s}_e(\lambda_e(v_e)v_e)+ (v_e^{\ast}-\lambda_e(v_e)v_e)\overline{t}_e(v_e).
\end{align*}
Summing over all $e\in E$ proves the lemma.
\end{proof}

\begin{prop}\label{pro:upper_bound_with_positive_distributions}
Let $(G, \mathbf{D}, \mathbf{c})$ be a transportation game with $\mathbf{D}$ positive demand distributions
and $\mathbf{c}\in \mathcal{C}$ with $\mathcal{C}$ a set of positive polynomial cost functions as in
\eqref{eqn:poly_cost_function}. Let $\mu_e(\cdot)$, $\phi_e(\cdot)$ and $\eta_e(\cdot)$ be defined as in Definition~\ref{def:aux_cost-functions} for each $c_e(\cdot)$, such that \eqref{eqn:poly_function_restriction}
is satisfied. Then
\[
 PoA(G, \mathbf{D}, \mathbf{c})\leq \gamma(\mathcal{C}, \mathbf{D}).
\]
\end{prop}

\begin{proof}
Let $\mathbf{p},\mathbf{p}^{\ast} \in \Omega$ be respectively the UE-SD and SO-SD strategies, with $\mathbf{v}$
and $\mathbf{v}^{\ast}$ as the corresponding mean link flows. Then
\begin{align*}
T(\mathbf{p}^{\ast})&=\sum_{e\in E}\mathbb{E}\,[c_e(V_e^*)V_e^*]
  \geq\sum_{e\in E} \underline{s}_e(v_e^{\ast})       & (\Leftarrow \eqref{eqn:expected_link_cost_bound0}) \\
&\geq\sum_{e\in E}\left(\underline{s}_e(\lambda_e(v_e)v_e)+
  (v_e^{\ast}-\lambda_e(v_e)v_e)\overline{t}_e(v_e)\right)  & (\Leftarrow \mbox{Lemma}\ \ref{lem:Se_lb}) \\
&=\sum_{e\in E}\left(\underline{s}_e(\lambda_e(v_e)v_e)+ v_e^{\ast}
  \overline{t}_e(v_e)- \lambda_e(v_e)v_e\overline{t}_e(v_e)\right) \\
&\geq\sum_{e\in E}\left(\underline{s}_e(\lambda_e(v_e)v_e)+ v_e
  \underline{t}_e(v_e)- \lambda_e(v_e)v_e\overline{t}_e(v_e)\right). & (\Leftarrow
   \mbox{Lemma}\ \ref{lem:weaker-UE-SD})
\end{align*}
Rewriting the last line of above leads to
\begin{align*}
T(\mathbf{p}^{\ast})&\geq\sum_{e\in E}\left(\mu_e(v_e)+\phi_e(v_e)-\eta_e(v_e)\lambda_e(v_e)\right)
  \overline{s}_e(v_e) \\
&\geq\frac{1}{\gamma(\mathcal{C},\mathbf{D})} \underset{e\in E}\sum\overline{s}_e(v_e) \geq \frac{1}{\gamma(\mathcal{C},\mathbf{D})}\underset{e\in E}\sum \mathbb{E}\,[c_e(V_e)V_e],
\end{align*}
which implies the proposition.
\end{proof}

The above proposition upper bounds the PoA by $\gamma(\mathcal{C},\mathbf{D})$. Next we show how to
compute the $\gamma$ value. Let $\mathcal{C}_m$ be the set of polynomial functions with nonnegative
coefficients and degree at most $m\in \mathbb{Z}_+$. Let $\widetilde{\mathcal{C}}_m$ be
the subset of $\mathcal{C}_m$ consisting of only one term, namely $\widetilde{\mathcal{C}}_m
=\{bx^j:b\geq 0, j=0,\ldots, m\}$. The following lemma shows the PoA with cost functions
in $\mathcal{C}_m$ is bounded by $\gamma(\widetilde{\mathcal{C}}_m,\mathbf{D})$.

\begin{lem}\label{c-cm-poly}
Let $\mathcal{I}_m =\{(G, \mathbf{D}, \mathbf{c}): \mathbf{c}\in \mathcal{C}_m\}$. Then
\[
\underset{(G, \mathbf{D}, \mathbf{c})\in \mathcal{I}_m}\sup \mbox{PoA}(G, \mathbf{D}, \mathbf{c})
\leq \gamma(\widetilde{\mathcal{C}}_m,\mathbf{D})
\]
\end{lem}
\begin{proof}
An arbitrary instance $(G, \mathbf{D}, \mathbf{c})$ with link cost functions in $\mathcal{C}_m$ can be transformed into an equivalent instance with link cost functions in  $\widetilde{\mathcal{C}}_m$ by replacing any edge $e\in E$ with link cost $c_e(x)=\sum_{j=0}^{m}b_{ej}x^j$ with a directed path consisting of $m+1$ links with the $j$-th link processing the cost $\widetilde{c}_{e,j}(x)=b_{ej}x^{j-1}$.
\end{proof}

\textbf{Remark.} A similar lemma can be found in \Citep{Roughgarden2005-book} for calculating the anarchy
value of polynomial cost functions in the deterministic models.

\begin{thm}\label{thm:compute_PoA_poly}
Let $\mathbf{D}$ be a vector of positive-valued random variables and satisfy
\begin{equation}\label{eqn:restriction_positive dist_m}
\overline{\theta}^{(m)}< (m+1)(1/m)^{m/(m+1)}, \ m\in \mathbb{Z}_+.
\end{equation}
Then
\begin{equation}\label{eqn:Gamma_m}
\gamma(\widetilde {\mathcal{C}}_m,\mathbf{D})=\max_{0\leq j\leq m} \left(\frac{1}{\overline{\theta}^{(j+1)}}-\frac{j}{j+1}
\frac{\overline{\theta}^{(j)}}{\overline{\theta}^{(j+1)}}
\left(\frac{\overline{\theta}^{(j)}}{j+1}\right)^{\hspace{-5pt}1/j}\,\right)^{\hspace{-5pt}-1}.
\end{equation}
\end{thm}

\begin{proof}
For any $c_e(\cdot)\in \widetilde {\mathcal{C}}_m$ with $c_e(x) =b_{ej}x^j$, we have
\[
\begin{array}{ll}
\displaystyle \overline{t}_e(v_e)=b_{ej} \overline{\theta}^{(j)}v_e^j, &\underline{t}_e(v_e)=b_{ej}v_e^j,\\
\displaystyle \overline{s}_e(v_e)=b_{ej} \overline{\theta}^{(j+1)}v_e^j,& \underline{s}_e(v_e)=b_{ej}v_e^{j+1},\\ \displaystyle \underline{s}_e'(v_e)=b_{ej}(j+1)v_e^j,
\end{array}
\]
from which we obtain
\begin{align*}
\lambda_e(v_e) &=({\overline{\theta}^{(j)}}/({j+1}))^{1/j},\\
\mu_e(v_e)&=\frac{1}{\overline{\theta}^{(j+1)}}
   \left(\frac{\overline{\theta}^{(j)}}{j+1}\right)^{\hspace{-5pt}1+1/j},\\
\phi_e(v_e)&=1/{\overline{\theta}^{(j+1)}},\\
\eta_e(v_e)&={\overline{\theta}^{(j)}}/{\overline{\theta}^{(j+1)}},
\end{align*}
which are all independent of $v_e$. If condition \eqref{eqn:poly_function_restriction} is satisfied, then
we have
\begin{align*}
\gamma(c_e,\mathbf{D})=&(\mu_e(v_e)+\phi_e(v_e)-\eta_e(v_e)\lambda_e(v_e))^{-1}\\
=&\left(\frac{1}{\overline{\theta}^{(j+1)}}-\frac{j}{j+1}
\frac{\overline{\theta}^{(j)}}{\overline{\theta}^{(j+1)}}
\left(\frac{\overline{\theta}^{(j)}}{j+1}\right)^{\hspace{-5pt}1/j}\,\right)^{\hspace{-5pt}-1},
\end{align*}
which implies \eqref{eqn:Gamma_m}. On the other hand, condition \eqref{eqn:poly_function_restriction}
requires that, for any $j=1,\ldots, m$,
\[
\frac{1}{\overline{\theta}^{(j+1)}}-\frac{j}{j+1}\frac{\overline{\theta}^{(j)}}
{\overline{\theta}^{(j+1)}}\left(\frac{\overline{\theta}^{(j)}}{j+1}\right)^{\hspace{-5pt}1/j}>0,
\]
which is equivalent to
\begin{equation}\label{constraint-general j}
\overline{\theta}^{(j)}< (j+1)(1/j)^{j/(j+1)}, \qquad \forall\ j=1,\ldots, m.
\end{equation}
It is routine to check that $(j+1)(1/j)^{j/(j+1)}$ is decreasing in $j$,
while $\overline{\theta}^{(j)}$ is increasing in $j$ according to \eqref{eqn:theta_factorization}. Therefore,
\eqref{constraint-general j} is implied by \eqref{eqn:restriction_positive dist_m}.
\end{proof}

The applicability of Theorem~\ref{thm:compute_PoA_poly} depends on satisfaction of
\eqref{eqn:restriction_positive dist_m}. Let us consider some practical values of $m$ in
\eqref{eqn:restriction_positive dist_m}: $m=2,3,4$, as the highest power of a link cost function is
seldom greater than 4 and usually the first four moments are studied in practice. Table~\ref{tab:D_region_positive}
lists the applicable ranges of these values.

\begin{table}[htbp]
\centering
 \begin{tabular}{lcl}
  \toprule
  Degree & Variability\\
  \midrule
$m=2$ & $\overline{\theta}^{(2)}< 1.889$\\
$m=3$ & $\overline{\theta}^{(3)}< 1.754$\\
$m=4$ & $\overline{\theta}^{(4)}< 1.649$\\
  \bottomrule
 \end{tabular}
 \caption{Applicable ranges of variability}\label{tab:D_region_positive}
\end{table}

Consider for example the uniform distribution $U[a,b]$ in Table~\ref{tab:D_region_positive}. The results of the
upper bounds on $b/a$, denoted by $[b/a]_{\max}$, are displayed in Table~\ref{tab:applicable_uni_dist} for normalized $a=1$.
\begin{table}[htbp]
\centering
 \begin{tabular}{lrl}
  \toprule
  Degree & $[b/a]_{\max}$\\
  \midrule
$m=2$ & $+\infty$ \\
$m=3$ & $14.241$ \\
$m=4$ & $3.556$ \\
  \bottomrule
 \end{tabular}
 \caption{Applicable uniform distribution $U[1,b]$}\label{tab:applicable_uni_dist}
\end{table}
As can seen, for link cost functions with the degree at most $m=2$, our upper bound on the PoA \eqref{eqn:Gamma_m} is applicable for
any positive-valued uniform distributions. For the cases $m=3,4$, our upper bounds on the PoA
are applicable to the uniform distributed demands with $b/a$ no more than $14.241$ and $3.556$, respectively.

\smallskip\textbf{Remark.} The upper bound of the PoA in Theorem~\ref{thm:compute_PoA_poly} is a generalization
of that provided by \Citep{Roughgarden2002} for deterministic models and the bound is tight when the demands return
to being deterministic. In fact, substituting $\overline{\theta}^{(j)}=1$ for every integer $0\leq j\leq m$ in
\eqref{eqn:Gamma_m}, the upper bound becomes
\[
\gamma(\widetilde {\mathcal{C}}_m, \mathbf{D})= \left(1-m(m+1)^{-(m+1)/m}\right)^{-1},
\]
which matches the tight upper bound of the PoA in deterministic models.

\subsection{Polynomial cost functions and normal distributions}

In this section, we still work on polynomial link cost functions \eqref{eqn:poly_cost_function}, but with
demands following the normal distributions, which are widely used in the literature to simulate
traffic demands, especially for a large mean or relatively small variance, although negative tails are contained \Citep{Watling2005, Asakura}.

Note that any path flow follows a normal distribution since it is a fraction of a normal distribution. Also any
link flow follows a normal distribution as it is the sum of some independent random variables of normal
distributions, i.e., $V_e \sim N(v_e, \sigma_e^2)$ for $e\in E$, with $v_e$ and $\sigma_e^2$ satisfying
\eqref{eqn:affine_implied_relation}. The $m$-th moment of the link flow on $e\in E$ can be written as a
function of the mean and variance of the link flow
\begin{equation}\label{eqn:V_e^m_normal}
\mathbb{E}[V_e^m]=\sum_{r=0,\,r=\textrm{even}}^m \dbinom{m}{r}(\sigma_e)^r (v_e)^{m-r}(r-1)!!, \quad
\forall\ e\in E,
\end{equation}
where $m\in \mathbb{N}$ is the power degree, $(r-1)!!$ is the double factorial of $r-1$, i.e., $(r-1)!!=(r-1)(r-3)\cdots 1$ (if $r$ is even) with the understanding that $(-1)!!=1$, and $\dbinom{m}{r}={m!}/({(m-r)!r!})$ is a binomial
coefficient. Our next lemma bounds the $m$-th moment of the link flow with functions of the mean link flow.
Let
\[
\ell_{m}=\sum_{r=0,\,r=\textrm{even}}^m \dbinom{m}{r}
\left(\frac{\underline{\theta}^2}{n}\right)^{r/2}(r-1)!!.
\]

\begin{lem}\label{lem:Ve_moment_normal}
Given a transportation game $(G,\mathbf{D},\mathbf{c})$ with $\mathbf{D}$ satisfying normal distributions, we
have
\[
\ell_{m}v_e^m \leq \mathbb{E}[V_e^m]\leq \overline{\theta}^{(m)} v_e^m, \ \forall\ e\in E.
\]
\end{lem}
\begin{proof}
Applying inequalities \eqref{eqn:sigma_lb} in \eqref{eqn:V_e^m_normal}, we obtain
\[
\mathbb{E}[V_e^m]\geq \sum_{r=0,\,r= \textrm{even}}^m \dbinom{m}{r}
\left(\frac{\underline{\theta}^2}{n}\right)^{r/2} (v_e)^m(r-1)!!, \  e\in E.
\]
On the other hand, observe that
\[
{\theta}_i^{(m)}=\sum_{r=0,\,r= \textrm{even}}^m \dbinom{m}{r}({\theta_i})^r(r-1)!!,\  \forall\ i\in I,
\]
which implies
\[
\overline{\theta}^{(m)}=\sum_{r=0,\,r=\textrm{even}}^m \dbinom{m}{r}(\overline{\theta})^r (r-1)!!,
\]
which together with \eqref{eqn:sigma_ub} implies the second inequality in the lemma.
\end{proof}

Compared with Lemma~\ref{lem:moment_Ve}, the upper bound on $\mathbb{E}[V_e^m]$ in
Lemma~\ref{lem:Ve_moment_normal} remains the same, while the lower bound is improved since $\ell_m>1$.
For polynomial link cost functions \eqref{eqn:poly_cost_function}, we will still use $\overline{t}_e$ and
$\overline{s}_e$ to upper bound the expected link flow functions and expected link total cost. However,
in order to establish the corresponding lower bounds, we use the following two new functions
$\underline{\tilde{t}}_e$ and $\underline{\tilde{s}}_e$ based on Lemma~\ref{lem:Ve_moment_normal} to replace
$\underline{t}_e$ and $\underline{s}_e$ for any $e\in E$:
\[
\underline{\tilde{t}}_e(v_e)=\sum_{j=0}^{m} b_{ej}\ell_{j}v_e^j, \quad \mbox{ and } \quad
\underline{\tilde{s}}_e(v_e)=\sum_{j=0}^{m} b_{ej} \ell_{j+1}v_e^{j+1}
\]
As usual, we use $\underline{\tilde{s}}_e'(\cdot)$ to denote the derivative of $\underline{\tilde{s}}_e(\cdot)$.

It is straightforward to check that, with positive-valued demand distributions replaced by the normal distributions and,
correspondingly,
with $\underline{t}_e$ and $\underline{s}_e$ replaced by $\underline{\tilde{t}}_e$ and $\underline{\tilde{s}}_e$
in Section~\ref{sec:poly}, Lemma~\ref{lem:weaker-UE-SD}, Lemma~\ref{lem:Se_lb} and
Proposition~\ref{pro:upper_bound_with_positive_distributions}, and hence Lemma~\ref{c-cm-poly} still hold. More
importantly, with the new functions of $\underline{\tilde{t}}_e$ and $\underline{\tilde{s}}_e$ for
dealing with normal distributions, the value of $\gamma(\widetilde{\mathcal{C}}_m)$ in
Lemma~\ref{c-cm-poly} will be different for normal demand distributions as shown in the following
new theorem, as compared with Theorem~\ref{thm:compute_PoA_poly}.

\begin{thm}\label{thm:upper_bound_value_for_normal}
Let $\mathbf{D}$ be a vector of the normal distributions and satisfy
\begin{equation}\label{eqn:restriction-normal-j}
\underset{1\leq j\leq m}\min \left\{ \ell_j-\frac{\overline{\theta}^{(j)} j}{(j+1)}\left(\frac{\overline{\theta}^{(j)}}{\ell_{j+1}(j+1)}\right)^{\hspace{-5pt}1/j} \right\} >0.
\end{equation}
Then
\begin{equation}\label{eqn:PoA_normal}
\gamma(\widetilde{\mathcal{C}}_m,\mathbf{D})=\max_{0\leq j\leq m}
\left(\frac{\ell_j}{\overline{\theta}^{(j+1)}}-\frac{\overline{\theta}^{(j)} j}
{\overline{\theta}^{(j+1)}(j+1)}\left(\frac{\overline{\theta}^{(j)}}
{\ell_{j+1}(j+1)}\right)^{\hspace{-5pt}1/j}\,\right)^{-1}.
\end{equation}
\end{thm}

\begin{proof}
With $\underline{t}_e$ and $\underline{s}_e$ replaced by $\underline{\tilde{t}}_e$ and
$\underline{\tilde{s}}_e$, we have
\begin{align*}
&\lambda_e=\left({\overline{\theta}^{(j)}}/({\ell_{j+1}(j+1))}\right)^{1/j},\\
&\mu_e=\left({\overline{\theta}^{(j)}}/{\ell_{j+1}(j+1)}\right)^{1/j}\cdot
{\overline{\theta}^{(j)}}/({\overline{\theta}^{(j+1)}(j+1)}) ,\\
&\phi_e= {\ell_j}/{\overline{\theta}^{(j+1)}},\\
& \eta_e={\overline{\theta}^{(j)}}/{\overline{\theta}^{(j+1)}}.
\end{align*}
Hence
\[
\mu_e+\phi_e-\eta_e\lambda_e=
\frac{\ell_j}{\overline{\theta}^{(j+1)}}-\frac{\overline{\theta}^{(j)} j}{\overline{\theta}^{(j+1)}(j+1)}\left(\frac{\overline{\theta}^{(j)}}{\ell_{j+1}(j+1)}\right)^{1/j}.
\]
The remaining proof is very much the same as that for Theorem~\ref{thm:compute_PoA_poly}.
\end{proof}

As in Section~\ref{sec:poly}, let us use numerical examples to illustrate the applicability of
Theorem~\ref{thm:upper_bound_value_for_normal} due to condition \eqref{eqn:restriction-normal-j}.
For simplicity, let $\theta=\underline{\theta}=\overline{\theta}$.
Figure~\ref{fig:max_theta_normal} illustrates the maximum applicable $\theta$ with a given $n$ for polynomial cost
functions with highest degree $m=2,3,4$. With an increasing $n$, the applicable region of $\theta$ narrows down
dramatically at the beginning and then remains almost constant when $n$ becomes greater than 5. If $n=1$, the
applicable $\theta$ can go to infinity for $m=2,3,4$. If $n=2$, the $\theta$ can go up to infinity for $m=2,3$ and
$0.77$ for $m=4$. If $n=5$, the applicable regions of $\theta$ are less than $1.476$, $0.670$ and $0.394$ for $m=2,3,4$,
respectively.

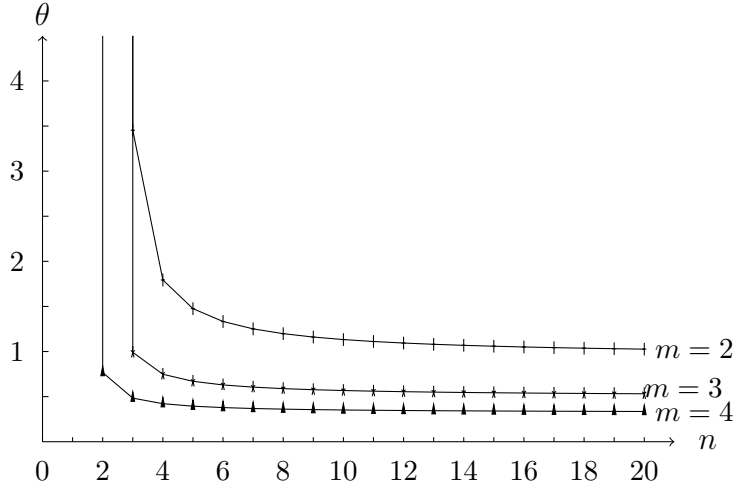
\begin{figure}
\centering
\begin{tikzpicture}
[xscale=0.4,yscale=1.2, domain=0:20]
\draw[->] (0,0) -- (21,0) node[right=5pt] {$n$};
\draw[->] (0,0) -- (0,4.5) node[above] {$\theta$};
\foreach \x/\textx  in {0,2,4,6,8,10,12,14,16,18,20}
{\draw (\x,0) -- (\x, 0.05) node [below=5pt]{\small $\x$};}
\foreach \x/\textx  in {1,3,5,7,9,11,13,15,17,19}
{\draw (\x,0) -- (\x, 0.05) ;}
\foreach \y/\texty  in {1,2,3,4}
{\draw (0,\y) -- (0.2,\y) node [left=5pt]{\small $\y$};}
\foreach \y/\texty  in {0.5,1.5,2.5,3.5}
{\draw (0,\y) -- (0.2,\y) ;}
\draw plot [mark=+] coordinates %
{ (3,3.45372)
        (4,	1.79357)
        (5,	1.47565)
        (6,	1.33306)
        (7,	1.25127)
        (8,	1.19804)
        (9,	1.16058)
        (10,1.13276)
        (11,1.11127)
        (12,1.09417)
        (13,1.08024)
        (14,1.06867)
        (15,1.0589)
        (16,1.05055)
        (17,1.04333)
        (18,1.03702)
        (19,1.03145)
        (20,1.02652)}node[right]{\small $m=2$};
\draw plot [mark=star] coordinates %
{(3,0.991206)
           (4,0.748085)
           (5,0.66982)
           (6,0.629806)
           (7,0.605317)
           (8,0.588738)
           (9,0.576752)
           (10,0.567676)
           (11,0.560563)
           (12,0.554836)
           (13,0.550125)
           (14,0.546181)
           (15,0.542831)
           (16,0.53995)
           (17,0.537446)
           (18,0.535249)
           (19,0.533305)
           (20,0.531575)
};
\node at (21.3,0.6){\small $m=3$};
\draw plot [mark=triangle*] coordinates %
        {(2,0.766966)
           (3,0.483332)
           (4,0.421674)
           (5,0.394156)
           (6,0.378492)
           (7,0.368363)
           (8,0.361268)
           (9,0.356021)
           (10,0.351981)
           (11,0.348774)
           (12,0.346167)
           (13,0.344005)
           (14,0.342183)
           (15,0.340628)
           (16,0.339283)
           (17,0.33811)
           (18,0.337077)
           (19,0.336161)
           (20,0.335342)
        }node[right]{\small $m=4$};
\draw (2,0.766966)--(2,4.5);
\draw (3,3.45372)--(3,4.5);
\draw (3,0.991206)--(3,4.5);
\end{tikzpicture}
\caption{Maximum applicable $\theta$ for normal distributions}
\label{fig:max_theta_normal}
\end{figure}
The next example shows the tightness of the upper bound.

\begin{eg}\label{eg:spigou-normal-poly}
Consider the two-link network in Figure~\ref{fig:two_link_network_of_normal_demand}. Assume the single
demand follows the Normal Distribution $D\sim N(d,\sigma^2)$.

\begin{figure}[H]
\centering
\begin{tikzpicture}
\draw (0,0) arc (0:360:4 and 1);
\draw [->] (-0.3,0.38)--(-0.2,0.3);
\draw [->] (-0.3,-0.38)--(-0.2,-0.3);
\node at (0,0) [circle,draw,fill=white] {t};
\node at(-8,0)[circle,draw,fill=white] {s};
\node at(-4, -1.4){$c_2(x)=x^j$};
\node at(-4, 1.4){$c_1(x)=\mathbb{E}[D^j]$};
\end{tikzpicture}
\caption{Two-link network}
\label{fig:two_link_network_of_normal_demand}
\end{figure}
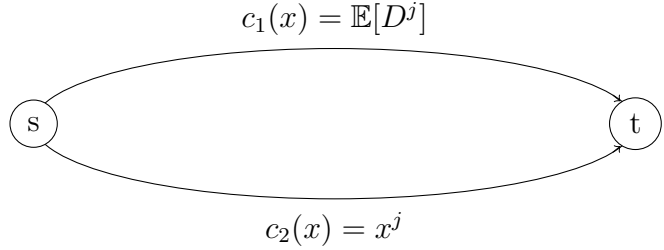

Define $g_j= \sum_{r=0,\,r=\textrm{even}}^j \dbinom{j}{r} \theta^r (r-1)!! $, where as before it is understood that $(-1)!! =1$.
Then $\mathbb{E}[D^j]=g_j d^j$. As the expected total cost on the lower link is never greater
than the upper link, strategy $\mathbf{p}^T=(0,1)^T$ is a UE-SD. We can calculate
$$
T(\mathbf{p})=\mathbb{E}[D^{j+1}]=g_{j+1}d^{j+1}.
$$
Let $\mathbf{p}^\ast=(p_1^\ast,p_2^\ast)^T$ be the the SO-SD strategy, which minimizes the
expected total cost
$$
T(\mathbf{p}^\ast)=p_1^\ast g_j d^{j+1}+(p_2^\ast)^{j+1}g_{j+1}d^{j+1}.
$$
Hence $p_1^\ast=1-[g_j/(g_{j+1}(j+1))]^{1/j}$ and $p_2^\ast=[g_j/(g_{j+1}(j+1))]^{1/j}$, which lead to
$$
T(\mathbf{p}^\ast)=\left( 1-\frac{j}{j+1}\left(\frac{g_j}
{g_{j+1}(j+1)}\right)^{\hspace{-4pt}1/j} \right) g_jd^{j+1}.
$$
Thus
$$
\textrm{PoA}=\left(\frac{g_j}{g_{j+1}}-\frac{g_j j}{g_{j+1}(j+1)}\left(\frac{g_{j}}{g_{j+1}(j+1)}\right)^{\hspace{-4pt}1/j}\right)^{-1}.
$$
\end{eg}
Therefore, the upper bound in \eqref{eqn:PoA_normal} is tight in the following two cases:
\begin{itemize}
\item When $n=1$ and $\overline{\theta}=\underline{\theta}$,
we have $\overline{\theta}^{(j)}=\ell_j=g_j$. The lower bound in Example \ref{eg:spigou-normal-poly} matches our upper bound.
\item When $m=1$, the upper bound matches the upper bound of the PoA with linear cost function established
in Section~\ref{sec:affine}. So it is tight when $\overline{\theta}=\underline{\theta}$.
\end{itemize}

\section{Discussion}

In our study with polynomial cost functions, we have established in Theorem~\ref{thm:compute_PoA_poly} and
\ref{thm:upper_bound_value_for_normal} two upper bounds on the PoA for two different demand settings,
namely, positive-valued distributions and the normal distributions. Based on the tightness analysis, the upper bound for normal
distributions is tight in a more general case as compared with that for general positive-valued distributions. In the study
for the normal distributions, we used two addition parameters, $n$ and $\underline{\theta}$, and the corresponding upper bound
on the PoA returns to the same as that for general positive-valued distributions when $n\rightarrow \infty$ or $\underline{\theta}\rightarrow 0$.

Next we use a numerical example of polynomial link cost functions with $m=2$ and demands of the normal distributions to compare the upper
bounds of the PoA with different values of $n$, as shown in Figure~\ref{fig:PoA_compare}. For simplicity, we consider the case that all the O-D pairs have a common coefficient of demand variation. From Theorem \ref{thm:upper_bound_value_for_normal}, the PoA is bounded by $\gamma(\widetilde{\mathcal{C}}_2, \mathbf{D}')$, where $\mathbf{D}'$ is the vector of demands with $\theta=\overline{\theta}=\underline{\theta}$.

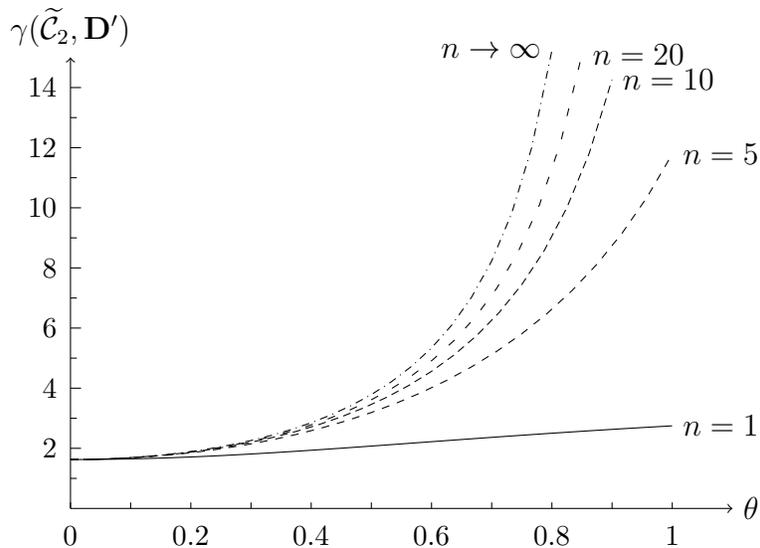
\begin{figure}[H]
\centering
\begin{tikzpicture}[xscale=8,yscale=0.4, domain=0:1]
\draw[->] (0,0) -- (1.1,0) node[right] {$\theta$};
\draw[->] (0,0) -- (0,15) node[above] {$\gamma(\widetilde{\mathcal{C}}_2, \mathbf{D}')$};
\foreach \x/\textx  in {0,0.2,0.4,0.6,0.8,1}
{\draw (\x,0) -- (\x, 0.2) node [below=5pt]{\small $\x$};}
\foreach \x/\textx  in {0.1,0.3,0.5,0.7,0.9}
{\draw (\x,0) -- (\x, 0.2) ;}
\foreach \y/\texty  in {2,4,6,8,10,12,14}
{\draw (0,\y) -- (0.5pt,\y) node [left=5pt]{\small $\y$};}
\foreach \y/\texty  in {1,3,5,7,9,11,13}
{\draw (0,\y) -- (0.3pt,\y) ;}
\draw  plot [domain=0:1](\x, {(9+27*\x^2)/((9+9*\x^2)-2*sqrt(3)*(1+\x^2)*sqrt((1+\x^2)/(1+3*\x^2)))})
node[right]{$n=1$};
\draw [dashed] plot [domain=0:1](\x, {(9+27*\x^2)/((9/5*(5+\x^2)-2*sqrt(15)*(1+\x^2)*
(sqrt((1+\x^2)/(5+3*\x^2)))) )})
node[right]{$n=5$};
\draw [densely dashed] plot [domain=0:0.9](\x, {(9+27*\x^2)/((9+0.9*\x^2)-2*sqrt(30)*(1+\x^2)*sqrt((1+\x^2)/(10+3*\x^2)))})
node[right]{$n=10$};
\draw [loosely dashed] plot [domain=0:0.85](\x, {(9+27*\x^2)/((9+0.45*\x^2)-4*sqrt(15)*(1+\x^2)*sqrt((1+\x^2)/(20+3*\x^2)))})
node[right]{$n=20$};
\draw [dashdotted] plot [domain=0:0.80](\x, {(9+27*\x^2)/(9-2*sqrt(3)*(1+\x^2)*sqrt(1+\x^2)))})
node[left]{$n\to \infty$};
\end{tikzpicture}
\caption{PoA comparison when $m=2$}
\label{fig:PoA_compare}
\end{figure}
As we can see, the upper bound is significantly better when the value of $n$ is small. As the upper bound of PoA with $n=1$ is tight,
the difference between the upper bound and lower bound is small when the demand variation is small (e.g., when $\theta<0.5$). The curve for $n\to \infty$ also illustrates the PoA with general positive-valued distributions.

The value $\overline{\theta}$ of the maximum variation of the demands is of vital importance in all the upper bounds. They increase
as it goes up and are all tight when it reduces to zero and hence the demands return to be deterministic. Therefore, our study
generalizes the upper bounds obtained by \Citet{Roughgarden2005-book} for deterministic demands. On the other hand,
with deterministic demands, the PoA with affine cost functions is bounded by 4/3, which indicates that the UE
is quite close to the system optimum. However, Example~\ref{eg:affine} shows that the PoA can be unbounded as $\overline{\theta}$
increases. Therefore, there is a fundamental difference between models of deterministic and stochastic demands.

Furthermore, as can been seen from \eqref{eqn:affine_implied_relation}, the variance of each link flow is affected by the number $n$
of O-D pairs whose paths share the link, and thus the variance depends on the network topology. The upper bounds we have obtained
with affine and polynomial cost functions for demands of the normal distributions also contain parameter $n$. Consequently,
unlike in the case of deterministic demands studied by \Citep{Roughgarden2002}, the network topology in general will affect the PoA
for stochastic demands. However, such an influence of the network topology is limited since we also have an upper bound without $n$
as shown in Figure~\ref{fig:PoA_compare} ($n\to \infty$).

\section{Conclusions}

In this study, we have presented a general equilibrium model for traffic games that take variation of the traffic demands into account.
The notion of mixed strategies is adopted in our models of user equilibrium and system optimum for stochastic demands to describe the travelers' and
coordinator's behaviors in a stochastic environment. The user equilibrium condition is reformulated as a VI problem, which enables us
to address the issue of existence and uniqueness of the equilibrium.

The PoA is bounded with affine and more general polynomial link cost functions respectively. For affine link cost functions,
a tight upper bound is established for general  demand distributions. For general polynomial link cost functions,
we bounded the PoA for two settings of the demand distributions, general positive-valued distributions and the normal distributions.
We have also demonstrated the tightness of the upper bounds under various special cases and presented numerical comparison among them.

We feel that the following issues are interesting to address and to extend our work in the future. Firstly,
as the expectation
of the travel costs and total cost are approximated by simple functions of the mean flows in this study, there is room for
generalization. Secondly, to improve the general upper bounds on the PoA, it would help to reformulate the user equilibrium condition
into another optimization problem. Thirdly, consider other specific demand distributions, such as the
log-normal distributions, to improve the general upper bound on the PoA.

\section*{Acknowledgements}

This work is funded by EPSRC, Science and Innovation Award (EP/D063
191/1). The authors wish to thank the support.

\bibliographystyle{elsarticle-harv}
\bibliography{ref}

\begin{thebibliography}{22}
\expandafter\ifx\csname natexlab\endcsname\relax\def\natexlab#1{#1}\fi
\expandafter\ifx\csname url\endcsname\relax
  \def\url#1{\texttt{#1}}\fi
\expandafter\ifx\csname urlprefix\endcsname\relax\def\urlprefix{URL }\fi

\bibitem[{Asakura and Kashiwadani(1991)}]{Asakura}
Asakura, Y., Kashiwadani, M., 1991. Road network reliability caused by daily
  fluctuation of traffic flow. European Transport: Highways \& Planning 19,
  73--84.

\bibitem[{Ashlagi et~al.(2006)Ashlagi, Monderer, and Tennenholtz}]{Ashlagi}
Ashlagi, I., Monderer, D., Tennenholtz, M., 2006. Resource selection games with
  unknown number of players. In: AAMAS '06 Proceedings of the fifth
  international joint conference on autonomous agents and multiagent systems.
  pp. 819--825.

\bibitem[{Bell and Cassir(2002)}]{Bell2002}
Bell, M. G.~H., Cassir, C., 2002. Risk-averse user equilibrium traffic
  assignment: an application of game theory. Transportation Research Part B:
  Methodological 36, 671--681.

\bibitem[{Bertsekas(1999)}]{Bertsekas1999}
Bertsekas, D.~P., 1999. Nonlinear Programming, 2nd Edition. Athena Scientific,
  Belmont, Massachusetts.

\bibitem[{Chau and Sim(2003)}]{Chau}
Chau, C.~K., Sim, K.~M., 2003. The price of anarchy for non-atomic congestion
  games with symmetric cost maps and elastic demands. Operations Research
  Letters 31, 327--334.

\bibitem[{Clark and Watling(2005)}]{Watling2005}
Clark, S., Watling, D., 2005. Modelling network travel time reliability under
  stochastic demand. Transportation Research Part B: Methodological 39~(2),
  119--140.

\bibitem[{Correa et~al.(2004)Correa, Schulz, and Stier-Moses}]{Correa2004}
Correa, J.~R., Schulz, A.~S., Stier-Moses, N.~E., 11 2004. Selfish routing in
  capacitated networks. Mathematics of Operations Research 29~(4), 961--976.

\bibitem[{Guo et~al.(2010)Guo, Yang, and Liu}]{Guo2010}
Guo, X., Yang, H., Liu, T.-L., 2010. Bounding the inefficiency of logit-based
  stochastic user equilibrium. European Journal of Operational Research
  201~(2), 463--469.

\bibitem[{Koutsoupias and Papadimitriou(1999)}]{Koutsoupias}
Koutsoupias, E., Papadimitriou, C., 1999. Worst-case equilibria. In:
  Proceedings of the 16th Annual Symposium on Theoretical Aspects of Computer
  Science. pp. 404--413.

\bibitem[{Lo et~al.(2006)Lo, Luo, and Siu}]{Lo2006}
Lo, H.~K., Luo, X., Siu, B.~W., 2006. Degradable transport network: travel time
  budget of travelers with heterogeneous risk aversion. Transportation Research
  Part B: Methodological 40~(9), 792--806.

\bibitem[{Myerson(1998)}]{Myerson1998}
Myerson, R.~B., 1998. Population uncertainty and poisson games. International
  Journal of Game Theory 27, 375--392.

\bibitem[{Nagurney(1998)}]{Nagurney1998}
Nagurney, A., 1998. Network Economics: A Variational Inequality Approach, 2nd
  Edition. Springer-Verlag New York, LLC.

\bibitem[{Perakis(2007)}]{Perakis}
Perakis, G., 2007. The ''price of anarchy'' under nonlinear and asymmetric
  costs. Mathematics of Operations Research 32~(3), 614--628.

\bibitem[{Roughgarden(2005)}]{Roughgarden2005-book}
Roughgarden, T., 2005. Selfish Routing and the Price of Anarchy. The MIT Press.

\bibitem[{Roughgarden and Tardos(2002)}]{Roughgarden2002}
Roughgarden, T., Tardos, E., 2002. How bad is selfish routing? J. ACM 49,
  236--259.

\bibitem[{Roughgarden and Tardos(2004)}]{Roughgarden2004}
Roughgarden, T., Tardos, E., 2004. Bounding the inefficiency of equilibria in
  nonatomic congestion games. Games and Economic Behavior 47~(2), 389--403.

\bibitem[{Schmidt(2003)}]{Schmidt2003}
Schmidt, K.~D., 2003. On the Covariance of Monotone Functions of a Random
  Variable. TU, Inst. für Mathematische Stochastik.

\bibitem[{Shao et~al.(2006)Shao, Lam, and Tam}]{Shao2006}
Shao, H., Lam, W., Tam, M., September 2006. A reliability-based stochastic
  traffic assignment model for network with multiple user classes under
  uncertainty in demand. Networks and Spatial Economics 6~(3), 173--204.

\bibitem[{Sheffi(1985)}]{Sheffi}
Sheffi, Y., 1985. Urban Transportation Networks: Equilibrium Analysis with
  Mathematical Programming Methods. Prentice-Hall.

\bibitem[{Sumalee and Xu(2011)}]{Sumalee2011}
Sumalee, A., Xu, W., 2011. First-best marginal cost toll for a traffic network
  with stochastic demand. Transportation Research Part B: Methodological
  45~(1), 41--59.

\bibitem[{Wardrop(1952)}]{Wardrop}
Wardrop, J.~G., 1952. Some theoretical aspects of road traffic research. ICE
  Proceedings: Engineering Divisions 1, 325--362.

\bibitem[{Zhou and Chen(2008)}]{Zhou2008}
Zhou, Z., Chen, A., 2008. Comparative analysis of three user equilibrium models
  under stochastic demand. Journal of Advanced Transportation 42, 239--263.

\end{thebibliography}







\end{document}